\documentclass[12pt]{article}
\pdfoutput=1

\usepackage{fullpage}
\usepackage[T1]{fontenc}
\usepackage{ae,aecompl}
\usepackage[dvips]{graphicx}
\usepackage{amssymb}
\usepackage{amsmath}
\usepackage{subfigure}
\usepackage[round,authoryear]{natbib}
\usepackage{color}
\usepackage{array}
\usepackage{rotating}
\usepackage{colortbl}
\usepackage{tikz}
\usetikzlibrary{patterns}

\definecolor{webgreen}{rgb}{0,0.4,0}
\definecolor{webbrown}{rgb}{0.6,0,0}
\definecolor{purple}{rgb}{0.5,0,0.25}
\definecolor{darkblue}{rgb}{0,0,0.7}
\definecolor{darkred}{rgb}{0.7,0,0}
\definecolor{darkgreen}{rgb}{0,0.7,0}
\usepackage[pdfborder=false]{hyperref}
\hypersetup{colorlinks,citecolor=darkred,filecolor=black,linkcolor=darkblue,urlcolor=webgreen,pdfpagemode=none,
pdfstartview=FitH}
\newcommand{\ignore}[1]{}
\DeclareMathOperator*{\argmax}{argmax}
\newtheorem{lemma}{{\sc Lemma}}
\newtheorem{remark}{{\sc Remark}}
\newtheorem{prop}{{\sc Proposition}}

\newtheorem{theorem}{{\sc Theorem}}
\newtheorem{defn}{{\sc Definition}}

\newtheorem{example}{{ Example}}

\newenvironment{proof}{\noindent {\bf \sl Proof\/}:\enspace}
{\hfill $\blacksquare{}$ \vspace{12pt}}
\usepackage{sectsty}
\sectionfont{\centering\normalfont\scshape}
\subsectionfont{\centering\normalfont}

\usepackage{cleveref}

\title{{\Large {\bf Selling two complementary goods}}~\thanks{We are grateful
to our supervisor Debasis Mishra for detailed comments and suggestions. We thank Albin Erlanson, Hemant Mishra, Arunava Sen, Salil Sharma, Swati Sharma for helpful comments. }}

\author{{\small
Komal Malik and Kolagani Paramahamsa}~\thanks{Komal Malik, Indian Statistical Institute, New Delhi, Email: \texttt{komal.malik18@gmail.com};
Kolagani Paramahamsa, Indian Statistical Institute, New Delhi, Email: \texttt{kolagani.paramahamsa@gmail.com}}}

\date{\small{\today}}

\begin{document}
\pagenumbering{roman}
\maketitle

\begin{abstract}
A seller is selling a pair of divisible complementary goods to an agent. The agent consumes the goods only in a specific {\it ratio} and freely disposes of excess in either good. The value of the bundle and the ratio are the agent's private information. In this two-dimensional type space model, we characterize the incentive constraints and show that the optimal (expected revenue-maximizing) mechanism is a {\it ratio-dependent posted price} or a {\it posted price} mechanism for a class of distributions. We also show that the optimal mechanism is a posted price mechanism when the value and the ratio are independently distributed.\\



\end{abstract}

\newpage

\pagenumbering{arabic}

\section{Introduction}
Complementary goods are often bundled and priced to various buyers' preferences. For instance, hardware and software are often priced together (personal computers with software or graphics card, gaming consoles with games). Products are offered together with different qualities of complementary services to suit buyers' preferences (exercise equipment with subscription services, automobiles with insurance). Other examples include a firm that needs two inputs in a particular ratio to produce a final good (coking coal and iron to produce steel); that is, the firm has a Leontief production function.

In all these cases, the {\it value} from consumption of the bundle and the consumption {\it ratio} is specific to the agent. A monopolistic seller who owns the complementary goods is selling to such an agent; what is the revenue-maximizing optimal pricing in these settings?

We analyze this question as a mechanism design problem. The agent's payoff is determined by a value, the desired ratio, and the quantity of the bundle she consumes. The value is interpreted as the payoff from consuming one unit of the primary good combined with the secondary good in the desired ratio. Both per unit value from the consumption of the bundle of goods and the ratio itself are the agent's private information. While one-dimensional pricing problems admit a posted price mechanism as optimal (\cite{Myerson81,Riley83}), multidimensional problems often lead to stochastic mechanisms that are difficult to implement. Our paper is motivated by finding simple optimal mechanisms in this multidimensional mechanism design setting.

For each report, a mechanism assigns quantities of both the goods, and payment to be made by the agent. Due to the revelation principle, we focus, without loss of generality, on direct mechanisms that are incentive compatible. An agent could potentially misreport both value and ratio dimensions. Dealing with incentive constraints in multidimensional mechanism design problems is difficult \citep{Manelli07, Carroll17}. We present this natural two-dimensional mechanism design model and show that a simple class of {\it non-wasteful} mechanisms are optimal under some conditions on the seller's beliefs over the agent's type. 

We show that a {\sc posted price} mechanism or a {\sc ratio-dependent posted price} mechanism is optimal. The former is a mechanism in which the seller offers one unit of one of the goods and the other good in the desired ratio at some fixed price. Each type gets the same bundle in the latter mechanism as in the former mechanism, but the price depends on the reported ratio. We first show that it is without loss of generality to focus on mechanisms in which allocations to any type are in the desired ratio; that is, the agent, after a truthful report, does not dispose of either of the goods that the mechanism allocates. This result allows us to use Myersonian techniques. We then characterize incentive compatible mechanisms and provide sufficient conditions over the seller's belief on the type-space for simple non-wasteful mechanisms to be optimal. Finally, we fully describe these mechanisms over the parameters of the problem. The price function that describes the optimal mechanisms is derived from the virtual valuation of the type's joint distribution.

\subsection{Related Literature}

\citet{armstrong_1996, rochet_chone_1998} analyze the standard multidimensional model with divisible goods, while \citet{mcafee_mcmillan_1988, manelli_vincent_2006}, among others, analyze the problem of indivisible goods. The optimal mechanism is known to be stochastic (allocation of the objects is randomized) for many distributions (\citet{hart_reny_2015,thanassoulis04}) in the case of indivisible goods. \citet{manelli_vincent_2006, devanur_haghpanah_psomas_2020, deb_bik} are among the papers that find sufficient conditions (on type distribution) under which deterministic mechanisms are optimal. Our model considers divisible complementary goods and finds sufficient conditions under which one of the goods is allocated to the maximum quantity available with the seller. This maximum quantity allocation is interpreted as a deterministic mechanism in the standard indivisible goods model (see \citet{pavlov_2011}).

\citet{devanur_haghpanah_psomas_2020} consider a model in which there are multiple copies of a good for sale. The agent derives a constant marginal `value' up to a `quantity' of the goods and no additional value beyond the desired quantity. The value and quantity are private information of the agent. They find conditions for deterministic mechanisms to be optimal and focus on the computational complexity of the problem. Our paper differs from theirs in that we consider a pair of heterogeneous goods with a privately known ratio of consumption while they consider homogeneous goods with privately known demand. The two optimization exercises are similar after we prove our first result of `non-wastefulness.' While they use the `utility' approach to show that there exists a deterministic optimal mechanism under some conditions,  we use the Myersonian approach to describe the optimal mechanism under a different set of conditions.  Although our conditions are stronger than theirs,  we fully describe the optimal mechanism, including the payments.

\citet{fiat_goldner_karlin_koutsoupias_2016}'s model has two-dimensional private type for a single object. One dimension is for `value,' which is constant up to a `deadline' and drops to zero beyond the deadline. The paper completely characterizes optimal mechanisms. The agent's utility in their model changes sharply beyond the deadline, while in our model (and in \citet{devanur_haghpanah_psomas_2020}'s model), the utility is a continuous function in allocations.

\section{The model}
A monopolistic seller is selling a pair of divisible goods to an agent. The seller has one unit each of the two goods, denoted by \textsc{good}$_1$ and \textsc{good}$_2$, and has no value for them. A consumption bundle for the agent is a tuple $(a_1,a_2,t)$, where $a_1,a_2 \in [0,1]$ are the allocation quantities of \textsc{good}$_1$ and \textsc{good}$_2$, respectively, and $t \in \mathbb{R}$ is the transfer - the amount paid by the agent. 

The utility derived by agent of type $(v,k)$ from an outcome $(a_1,a_2,t)$ is given by,
\begin{displaymath}
U_{(v,k)}(a_1,a_2,t) := v \min \left\{\frac{a_1}{k},a_2\right
\} - t.
\end{displaymath}

The agent treats the goods as perfect complements; that is, any two allocations $(a_1,a_2)$ and $(a'_{1},a'_{2})$ with $\min \{\frac{a_1}{k},a_2\} = \min \{\frac{a'_1}{k},a'_2\}$ are payoff equivalent, where $k \in K \equiv (0,1]$ is the ratio of quantities of \textsc{good}$_1$ and \textsc{good}$_2$ that the agent demands.\footnote{Using such Leontief utility function to model complementary goods is standard in the literature \citep{telser_1979}.} The agent uses a standard quasilinear utility function to compare combinations of such bundles and the payment. This is captured by a per-unit value $v$ drawn from $V \equiv [0,1]$. It is worth noting that, given an allocation and a payment, the utility is a continuous function of the ratio.

In our model, both $v$ and $k$ are private information of the agent, therefore the agent has a ``type" $(v,k) \in V\times K$. \textsc{good}$_2$ is the {\it primary} good, whereas \textsc{good}$_1$ is its {\it complement} which is always consumed lesser in quantity than the former as $k \in (0,1]$.

Consider a buyer in the market for a personal computer who demands $x$ units of CPU and $y$ units of GPU to process an application, and her willingness to pay is $v$ for this bundle. 
The buyer's willingness to pay for a combination of $\frac{x}{3}$ units of CPU and $\frac{y}{2}$ units of GPU is then $\frac{v}{3}$ since the extra allocation in GPU units goes unused (since the specific application requires a CPU to GPU ratio of $x:y$).  Another buyer (a gamer) may have a different value and demand a ratio lower than $x:y$. Differentially priced bundles catering to buyers with different preferences (both consumption ratio and willingness to pay) are commonplace in these markets.  Our paper sheds light on how incentive constraints play a role in setting revenue-maximizing prices.

We assume that the random variables $v,k$ follow a joint distribution function $G$ with strictly positive density function $g$. We use $g_v$ and $g_k$ to denote marginal density functions of $V$ and $K$, respectively. The conditional density of $v$ given $k$ is denoted by $g(v|k)$.

\section{Optimal Mechanism}
        
An allocation function $f: V \times K \rightarrow [0,1]^2$ and a payment function $p: V \times K \rightarrow \mathbb{R}$ define a direct mechanism $(f,p)$. For any allocation function $f$, we use subscript notations $f_1$ and $f_2$ to denote allocations corresponding to \textsc{good}$_1$ and \textsc{good}$_2$, respectively. Standard revelation principle argument implies that we can focus, without loss of generality, on incentive compatible direct mechanisms. 

\begin{defn}
A mechanism $(f,p)$ is {\bf incentive compatible (IC)} if for all $(v,k), (v',k') \in V \times K$,
$$U_{(v,k)}(f(v,k),p(v,k)) \geq U_{(v,k)}(f(v',k'),p(v',k'))$$
\end{defn}

IC condition ensures that the agent has the incentive to report her type - both value and ratio - truthfully. 

\noindent {\bf Notation.} We use $(v,k) \rightarrow (v',k')$ to denote the incentive constraint for the type $(v,k)$ misreporting as type $(v',k')$. 

We also impose a participation constraint. The utility for every type of the agent is at least zero from participating in the mechanism.
 
\begin{defn}
A mechanism $(f,p)$ is {\bf individually rational (IR)} if for all $(v,k) \in V \times K$,                           
$$U_{(v,k)}(f(v,k),p(v,k)) \geq 0.$$
\end{defn}

The expected (ex-ante) revenue of a mechanism $(f,p)$ is given by
$$\Pi(f,p) := \int_{V\times K} p(v,k) dG(v,k).$$

We say that a mechanism $(f,p)$ is {\bf optimal} if
\begin{itemize}
\item $(f,p)$ is IC and IR,
\item and $\Pi(f,p) \geq \Pi(f',p')$ for any other IC and IR mechanism $(f',p')$.
\end{itemize}

\subsection{Simple Mechanisms}

Before our main results, we describe two {\it simple} mechanisms and show that they are IC and IR. 
These mechanisms are easy to describe since the primary good is allocated fully or not at all, and the complement is allocated in the desired ratio. We provide an example to show that the {\it simple} mechanisms are not always optimal. In the subsequent sections,  we show that these mechanisms are optimal under some conditions.

\begin{defn}
A mechanism $(f,p)$ is {\sc posted price} mechanism if there exists a $\rho^* \in [0,1]$ such that
\begin{displaymath}
(f(v,k),p(v,k)) = \left\{ \begin{array}{ll}
(0,0,0) & \textrm{if $v \le \rho^*$}\\
(k,1,\rho^*) & \textrm{otherwise.}
\end{array} \right.
\end{displaymath} 
\end{defn}

In a {\sc posted price} mechanism, there exists a price $\rho^*$ such that all the types whose value is less than $\rho^*$ get no good and pay nothing. A type $(v,k)$ with $v > \rho^*$ gets $k$ units of \textsc{good}$_1$, $1$ unit of \textsc{good}$_2$, and pays $\rho^*$ to the seller.

\begin{defn}
A mechanism $(f,p)$ is {\sc ratio-dependent posted price} mechanism if there exists a function $\psi: K \rightarrow V$ such that
for all $k' > k$,
\begin{align*}
\psi(k) &\leq \psi(k'),\\
\frac{k}{k'}\psi(k') &\leq \psi(k), and\\
(f(v,k),p(v,k)) &= \left\{ \begin{array}{ll}
(0,0,0) & \textrm{if $v \le \psi(k)$}\\
(k,1,\psi(k)) & \textrm{otherwise.}
\end{array} \right.
\end{align*} 
\end{defn}

\begin{figure}[!tbp]
  \centering
  \begin{minipage}[b]{0.35\textwidth}
    \includegraphics[width=\textwidth]{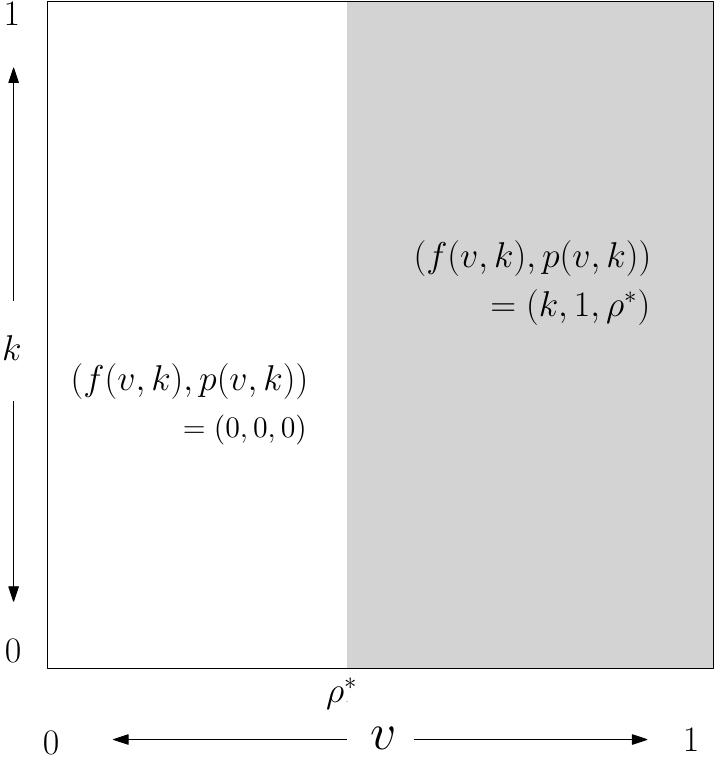}
    \caption{{\sc posted price}}
  \end{minipage}
  \hfill
  \begin{minipage}[b]{0.35\textwidth}
    \includegraphics[width=\textwidth]{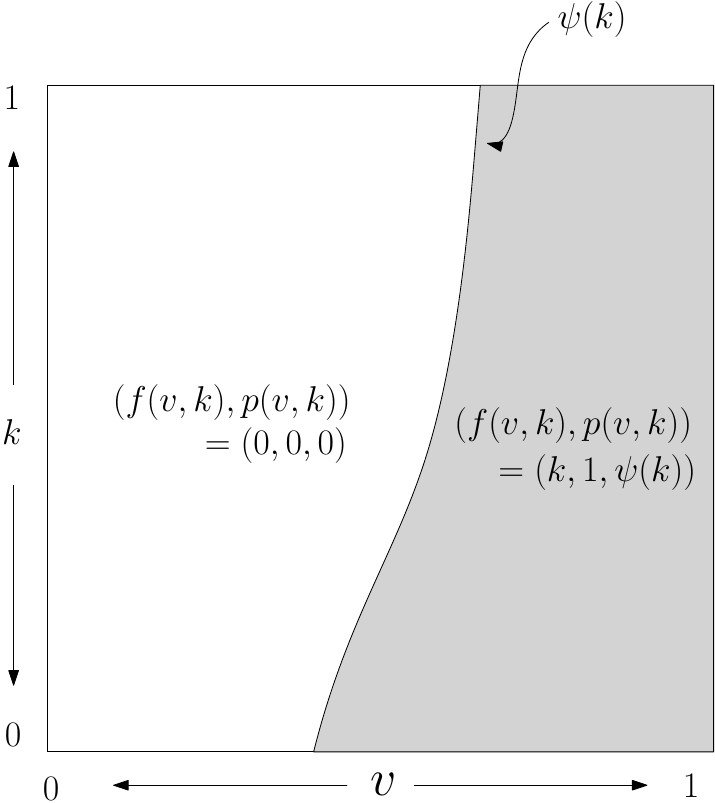}
    \caption{{\sc ratio-dependent posted price}}
  \end{minipage}
\end{figure}

For instance $\psi(k) = (\frac{1}{k+2})^{\frac{1}{k+1}}$ satisfies the conditions that defines a {\sc ratio-dependent posted price} mechanism, and this is not a {\sc posted price} mechanism. Observe that the {\sc posted price} mechanism is a special case of the {\sc ratio-dependent posted price} mechanism by setting $\psi(k) = \rho^*$ for all $k$. Our next proposition shows that the {\sc ratio-dependent posted price} mechanism is IC and IR. Therefore, this also proves that the {\sc posted price} mechanism is IC and IR. While the {\sc posted price} mechanism has a unique price in the menu, the {\sc ratio-dependent posted price} mechanism has a potentially infinite-sized menu.  Omitted proofs are relegated to the Appendix \ref{sec:appa}. 

\begin{prop}\label{prop:rdic}
A {\sc ratio-dependent posted price} mechanism is IC and IR.
\end{prop}

Following example illustrates that {\it simple} mechanisms are not always optimal.

\begin{example}\label{ex:1} \normalfont Consider a discrete type-space: $\{(v_1,k_1) = (1,0.25), (v_2,k_2) = (6,0.5), (v_3,k_3) = (1,1)\}$ with probabilities $10a,a,a$,  respectively\footnote{Type-space in $v$ dimension can be made continuous in this example without affecting its conclusion. Continuity in $k$ dimension is not important for our results, making this example valid. Here, $a = 1/12$ is the normalizing constant.}. Let the outcome $(a_i,b_i,t_i)$ be assigned to the type $(v_i,k_i)$. If we impose a constraint that $b_i \in \{0,1\}$, then the menu $\{(a_i,b_i,t_i)\}_{i\in\{1,2,3\}} = \{(0.25,1,1), (0.5, 1, 4), (0, 0, 0)\}$ yields the maximum revenue of $14a$ (proof in Appendix \ref{subsec:ex1}).

However, it is easy to verify that the menu: $\{(0.25,1,1), (0.5, 1, 4), (0.4, 0.4, 0.4)\}$ is IC, IR and generates a {\bf strictly} higher revenue of $14.4a$. Clearly, this is not a {\it simple} mechanism as $b_3 \notin \{0,1\}$.
\end{example}

\subsection{Main Results}

We impose the following restrictions on type distribution $G$ for our results.

\begin{defn}
A distribution $G$ satisfies {\sc Condition A} if for any $k$, $v(1-G(v|k))$ is strictly concave in $v$.
\end{defn}

This is a standard condition used in the literature in other settings~\citep{Che00,devanur_haghpanah_psomas_2020}. Given a ratio $k$, this condition may be interpreted as diminishing marginal revenue (see \cite{devanur_haghpanah_psomas_2020}). This is because, conditional on $k$, $v(1-G(v|k))$ is the revenue from charging price $v$ for one unit of the primary good and $k$ units of the secondary good.

Define
\begin{equation*}
\phi(v,k) := v - \frac{1-G(v|k)}{g(v|k)}.
\end{equation*}
{\sc Condition A} ensures that the solution to $\phi(v,k)g(v|k) = 0$ and $\phi(v,k)=0$ is the same and unique\footnote{Note that $\phi(0,k)g(0,k) < 0$ and $\phi(1,k)g(1,k) > 0$ for all $k$ and that continuity of $G$ ensures continuity of $\phi(v,k)g(v|k)$.  {\sc Condition A} implies that $\phi(v,k)g(v|k)$ strictly increasing in $v$ for every $k$. Since $g(v,k) > 0$, we have the uniqueness.}. For any $k$, we denote the unique value satisfying these equations by $\phi_k^{-1}(0)$. 
Observe that, conditional on $k$;  the optimal mechanism is posted price $\phi_k^{-1}(0)$. That is, if the ratio is public information, then the seller would post this price for one unit of primary good and $k$ units of secondary good since the optimization problem reduces to that of a single object model. 

\begin{defn}
A distribution $G$ satisfies {\sc Condition B} if it satisfies {\sc Condition A} and for all $k < k'$ the following is true,
$$ \phi_k^{-1}(0) > \phi_{k'}^{-1}(0).$$
\end{defn}
This condition states that the optimal price conditional on $k$ is decreasing in $k$.  In this case, maximizing revenue point-wise (for each $k$) implies allocating one unit of primary good and $k$ units of secondary good for a price $\phi_k^{-1}(0)$.  However, this mechanism is not incentive compatible. To see why - a type with lower ratio would deviate to a type with higher ratio and pay less for more of the secondary good(unused part can be discarded).  Using {\sc Condition A} and ironing techniques, we show that the optimal mechanism, in this case, is a {\sc posted price} mechanism.

\begin{theorem}\label{theo:last}
If $G$ satisfies {\sc Condition B}, then a {\sc posted price} mechanism is optimal.
\end{theorem}

Following is an example of a distribution that satisfies {\sc Condition B} and the corresponding
optimal mechanism.

\begin{example}\label{ex:2}\normalfont
 Consider a density function $g(v,k) = \frac{2}{3}(v+2k)$. We evaluate the conditional density to $g(v|k) = \frac{v+2k}{0.5+2k}$. From this we derive the virtual valuation to, $$\phi(v,k) = \frac{1.5v^2+4kv-2k-0.5}{v+2k}.$$

We can show that $\phi(v,k)$ is strictly increasing by first order condition,  this implies that {\sc Condition A} is satisfied. Also, $$\phi_{k}^{-1}(0) = \frac{-4k + \sqrt{16k^2+12k+3}}{3}$$ is decreasing implies that {\sc Condition B} holds. Therefore, the optimal mechanism for this distribution evaluates to,

\begin{align*}
\big(f(v,k),p(v,k)\big) = 
   \begin{cases}
   (0,0,0) \qquad v \leq \rho^*\\
   (k,1,\rho^*)\quad \text{otherwise}
   \end{cases}
\end{align*}

\text{where } $\rho^* = \argmax_{\rho} \rho \big(1-G_{v}(\rho)\big)$, this evaluates to $\rho^* = \frac{\sqrt{13}-2}{3}$. This generates a revenue of $0.2931$.

Now, suppose that $k$ is {\bf public} information. Then the optimal mechanism is to simply post a price $ \phi_k^{-1}(0)$ (for one unit of primary good and $k$ units of the secondary good) for each $k$. This mechanism generates an ex-ante expected revenue of $0.2933$. This observation illustrates that there is information rent to be paid in the ratio dimension as well.
\end{example}

\begin{defn}
A distribution $G$ is said to satisfy {\sc Condition B$^\prime$} if it satisfies {\sc Condition A} and for all $k < k'$ the following is true,
$$\frac{k}{k'}\phi_{k'}^{-1}(0) \leq \phi_k^{-1}(0) \leq \phi_{k'}^{-1}(0).$$
\end{defn}

This condition states that the optimal price conditional on $k$ is increasing in $k$ but boundedly. Observe that this condition is partially complementary to {\sc Condition B}, where $\phi^{-1}_k(0)$ is decreasing. We discuss both the conditions in detail in section \ref{sec:compare}.

\begin{theorem}\label{theo:one}
If $G$ satisfies {\sc Condition B$^\prime$}, then the following {\sc ratio-dependent posted price} mechanism is optimal,
\begin{align*}
\big(f(v,k),p(v,k)\big) = 
   \begin{cases}
   (0,0,0) \qquad v \leq \phi_k^{-1}(0)\\
   \left(k,1,\phi_k^{-1}(0)\right)\quad \text{otherwise}
   \end{cases}
\end{align*}
\end{theorem}

The proof proceeds by ignoring the ratio-related incentive constraints to maximize revenue. We then argue that such a mechanism is incentive compatible. The revenue from the optimal mechanism under {\sc Condition B$^\prime$} would be the same as ex-ante optimal revenue if the ratio were public information. This mechanism admits a two-part tariff implementation in which the seller charges $\rho_p := \lim_{k \to 0^+}\phi^{-1}_{k}(0)$ for the primary good, and the secondary good is priced at $\phi^{-1}_{k}(0) - \rho_p$ for $k$ units.

\begin{example}\label{ex:3}\normalfont
 Consider a density function $g(v,k) = \frac{v^k}{\ln{2}}$ whose conditional density evaluates to $g(v|k) = v^k(k+1)$. From this we derive the virtual valuation to, $$\phi(v,k) = v - \frac{1-v^{k+1}}{v^k(k+1)}.$$

We can show that $\phi(v,k)$ is strictly increasing by first order condition, and $\phi_{k}^{-1}(0) = (\frac{1}{k+2})^{\frac{1}{k+1}}$ implies $G$ satisfies {\sc Condition B$^\prime$}. Therefore, the optimal mechanism for this distribution, using Theorem \ref{theo:one},  evaluates to,

\begin{align*}
\big(f(v,k),p(v,k)\big) = 
   \begin{cases}
   (0,0,0) \qquad v \leq (\frac{1}{k+2})^{\frac{1}{k+1}}\\
   (k,1,(\frac{1}{k+2})^{\frac{1}{k+1}})\quad \text{otherwise}
   \end{cases}
\end{align*}

The seller can implement this mechanism by charging $\frac{1}{2}$ for the primary good and $(\frac{1}{k+2})^{\frac{1}{k+1}} - \frac{1}{2}$ for $k$ units of the secondary good.

\end{example}

\begin{remark} \normalfont 
Observe that replacing the outcome $(0,0,0)$ in either the {\sc posted price} or the {\sc ratio-dependent posted price} mechanism with an outcome $(0,1,0)$ generates the same revenue without violating IC and IR constraints. It may appear that the primary good is not required to screen the buyer since it is always allocated. However,  this is a consequence of optimization subject to our sufficient conditions. Example \ref{ex:1} illustrates the case where this is not true. 

The agent’s private information is in both value and ratio dimensions. In Example \ref{ex:2}, we illustrate that the seller pays rent in the ratio dimension. The seller uses value dimension in both the optimal mechanisms we describe. Therefore, this model cannot be reduced to a one-dimensional private information model.
\end{remark}

\subsection{Optimal Program}

In this section, we derive our main results. We can restrict the class of mechanisms to optimize over due to the following result.

\begin{prop}\label{prop:red}
For every IC and IR mechanism $(f,p)$ there exists another IC and IR mechanism $(f',p')$ such that
\begin{enumerate}
    \item $\Pi(f',p') = \Pi(f,p)$, and
    \item $f'_{1}(v,k) = k f'_{2}(v,k)$ for all $(v,k).$ - {\it{\bf non-wasteful allocation}}
\end{enumerate}
\end{prop}

Proposition \ref{prop:red} implies that, to find the optimal mechanism it is without loss of generality to focus on the class of mechanisms with the property that allocation of \textsc{good}$_1$ is $k$ times allocation of \textsc{good}$_2$. To prove this, we start with an arbitrary IC and IR mechanism $(f,p)$ and construct the desired form mechanism $(f',p')$ while keeping the revenue constant. $(f',p')$ is derived from $(f,p)$ by reducing the allocation of one of the goods so that the allocation ratio is as reported. The payments remain the same. 

These {\bf non-wasteful} mechanisms are denoted by,
$$\mathcal{M}:= \{(f,p): f_1(v,k) = kf_2(v,k)~\text{for all}~ (v,k)\}.$$

\noindent \underline{Note}:
 If $(f,p) \in \mathcal{M}$ then $\min \left\{\frac{f_1(v,k)}{k},f_2(v,k)\right
\} = f_2(v,k)$ for all $(v,k)$. In the next Proposition and rest of the paper we use the following fact without explicitly stating: for any $k$, $$U_{(v,k)}(f(v',k),p(v',k)) = vf_2(v',k) - p(v',k)~\text{for all}~v,v'.$$

This allows us to focus only on one of the allocation function components $f_2$, and deduce $f_1$ from it in the final step. However, this does not reduce the problem to a one-dimensional exercise as incentive constraints across the ratio dimension are crucial to the optimal program. The following result makes this point clear.

\subsubsection{Characterization of IC Mechanisms}

We characterize the IC mechanisms in the class $\mathcal{M}$.

\begin{prop}\label{prop:ICC}
$(f,p) \in \mathcal{M}$ is IC if and only if the following are true for any $(v,k)$,
\begin{enumerate}
\itemsep0em 
\item [(1)]$\displaystyle f_2(v,k) \leq f_2(v',k)$ for all $v' > v$,
\item [(2)]$\displaystyle p(v,k) = p(0,1) + vf_2(v,k) - \int_{0}^{v}f_2(t,k)dt$,
\item [(3)]$\displaystyle \int_{0}^{v}f_2(t,k')dt \leq \int_{0}^{v}f_2(t,k)dt$ for all $k' > k$,
\item [(4)]$\displaystyle \int_{0}^{v\frac{k}{k'}}f_2(t,k)dt \leq \int_{0}^{v}f_2(t,k')dt$ for all $k' > k$.
\end{enumerate}
\end{prop}

The conditions (1) and (2) in Proposition \ref{prop:ICC} correspond to IC constraints between two types on a horizontal line in the type-space (see Figure \ref{fig:icillus}). Mechanisms in $\mathcal{M}$ have the property of reducing the IC constraints on any horizontal line equivalent to that of a one-dimensional problem. This is the same as \citet{Myerson81}'s IC characterization when restricted to any $k$. However, Proposition \ref{prop:ICC} shows that some `vertical' and `diagonal' constraints are enough to guarantee the incentive compatibility of the mechanism. Condition (3) corresponds to the vertical constraints, while (4) corresponds to the diagonal constraints. The arrows in Figure \ref{fig:icillus} indicate the direction in which the incentive constraints need to be satisfied. 

Interestingly, these `local' constraints are enough to guarantee global incentive compatibility. Describing optimal mechanisms in multidimensional models is difficult partly because binding constraints cannot be pinned down (\citet{rochet_chone_1998}). However, we can do so in this model due to the nature of incentive constraints.

\begin{figure}[!hbt]
\centering
\includegraphics[width=2.7in]{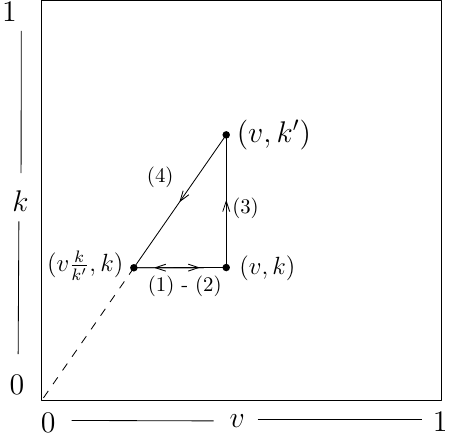}
\caption{IC Constraints}
\label{fig:icillus}
\end{figure}

We state two lemmas which we use in our further analysis.

\begin{lemma}\label{lem:zero}
If a mechanism $(f,p)$ is IC, then $p(0,k) = p(0,1)$ for all $k$.
\end{lemma}
\begin{proof}
For any $k$, $(0,1) \rightarrow (0,k)$ implies that $-p(0,1) \geq -p(0,k)$ while $(0,k) \rightarrow (0,1)$ implies that $-p(0,k) \geq -p(0,1)$.
\end{proof}

In line with other models in mechanism design, the following standard result holds in this setting too.

\begin{lemma}\label{lem:ir}
An IC mechanism $(f,p)$ is individually rational if and only if,
$$p(0,1) \leq 0.$$
\end{lemma}
\begin{proof}
Fix an IC mechanism $(f,p)$. Suppose that  $(f,p)$ is IR. Consider the type $(0,1)$. IR implies that $U_{(0,1)}(f(0,1),p(0,1)) \geq 0$, this simplifies to $p(0,1) \leq 0$. To show the other way, fix any $(v,k)$ and observe that $U_{(v,k)}(f(v,k),(v,k)) = v\min\{\frac{f_1(v,k)}{k},f_2(v,k)\}-p(v,k) \geq v\min\{\frac{f_1(0,k)}{k},f_2(0,k)\}-p(0,k) \geq -p(0,k) \geq 0$. The first inequality is from the incentive constraint $(v,k) \rightarrow (0,k)$, the second from the fact that allocation functions are non-negative. The third is true since $p(0,1) \leq 0$ implies that $p(0,k) \leq 0$ for all $k$ due to Lemma \ref{lem:zero}. 
\end{proof}

Using Lemma \ref{lem:zero}, Lemma \ref{lem:ir}, and Proposition \ref{prop:ICC}, 
the optimal program can be summarized as follows, 
%

\begin{samepage}
\begin{center}
{\bf Optimal Program}
\end{center}
\begin{align*}
\max_{f_2:V\times K \rightarrow [0,1]} \int_{0}^{1}\Bigg[ \int_{0}^{1} \phi(v,k) f_2(v,k)g(v|k) dv\Bigg] g_k(k)dk \tag{O}\label{eq:obj}\\
f_2(v,k) \leq f_2(v',k)~\text{for all}~v < v',k,\tag{C1}\label{eq:mon}\\
\int_{0}^{v}f_2(t,k')dt \leq \int_{0}^{v}f_2(t,k)dt~\text{for all}~v, k' > k,\tag{C2}\label{eq:icv}\\
\int_{0}^{v\frac{k}{k'}}f_2(t,k)dt \leq \int_{0}^{v}f_2(t,k')dt~\text{for all}~v, k' > k\tag{C3}\label{eq:icd},\\
\end{align*}
\end{samepage}

Notice that $f_2$ uniquely determines $f_1$ and $p$ by Propositions \ref{prop:red} and \ref{prop:ICC}, respectively. Therefore, these decision variables are suppressed in the above optimal program.

\subsubsection{Proof of Theorem \ref{theo:last}.}\label{subsec:proplast}
We solve for the optimal mechanism by ignoring the constraint (C3). We show that a {\sc posted price} mechanism is optimal for this reduced problem. We first prove the following Lemma towards this.

\begin{lemma}\label{lem:impmnt}
If $G$ satisfies {\sc Condition A} then for every mechanism $(f,p) \in \mathcal{M}$ that satisfies constraints (C1), (C2), then the mechanism $(f',p') \in \mathcal{M}$ defined by,
\begin{displaymath}
f_2'(v,k) = \left\{ \begin{array}{ll}
0 & \textrm{if $v \le 1-\int_{0}^{1}f_2(t,k)dt$}\\
1 & \textrm{otherwise.}
\end{array} \right.
\end{displaymath}
satisfies constraints (C1), (C2) and generates more(weakly) expected revenue than $(f,p)$.
\end{lemma}

Lemma \ref{lem:impmnt} implies that, without loss of generality, we can focus on mechanisms $(f,p) \in \mathcal{M}$ such that there exists $\rho(k)$ {\bf increasing} in $k$ and,
\begin{equation}\label{eq:theo}
f_2(v,k) = \left\{ \begin{array}{ll}
0 & \textrm{if $v \le \rho(k)$}\\
1 & \textrm{otherwise.}
\end{array} \right.
\end{equation}
$\rho$ is increasing because (C2) is satisfied in Lemma \ref{lem:impmnt}, and due to the definition of $f'$ in Lemma \ref{lem:impmnt}. We show that we can improve such a mechanism to a {\sc posted price} mechanism by ironing (Appendix \ref{subsec:ironing}).

\subsubsection{Proof of Theorem \ref{theo:one}.}
Ignoring the constraints (C1), (C2), and (C3), a point-wise maximization (for each $k$) of the objective function (\ref{eq:obj}) implies that the optimal allocation function $f_2$ is as in the statement of the theorem, since $\phi(v,k) \leq 0$ for all $(v,k)$ with $v \leq \phi_k^{-1}(0)$ and $\phi(v,k) > 0$ for all $(v,k)$ with $v > \phi_k^{-1}(0)$, due to {\sc Condition A}. {\sc Condition B$^\prime$} implies that this mechanism is indeed {\sc ratio-dependent posted price} mechanism. We have already shown this mechanism to be IC (Proposition \ref{prop:rdic}). Hence, the ignored constraints hold.

\subsubsection{On sufficient conditions of Theorems \ref{theo:last} and \ref{theo:one}}\label{sec:compare}

{\sc Condition B $\&$ B$^\prime$} are partially complementary in the sense that the former states that the function $\phi^{-1}_k(0)$ is decreasing while the latter says it is increasing (and rate of increase is bounded).

In the following example, {\sc Condition A} is satisfied but neither {\sc Conditions B} or {\sc B$^\prime$} hold.  We illustrate that the optimal mechanism is not in the class of {\it simple mechanisms}.

\begin{example} \label{ex:4} \normalfont
$V \times K \equiv [0,1]\times \{0.75,1\}$. The density $g$ is given by,
\begin{displaymath}
g(v,k) = \left\{ \begin{array}{ll}
0 & \textrm{if $(v < \frac{1}{2},~k = 0.75)$ or $(v<\frac{2}{3},~k=1)$}\\
100a & \textrm{if $v \geq \frac{1}{2},~k=0.75$}\\
a & \textrm{if $v \in [\frac{2}{3},\frac{3}{4}),~k=1$}\\
10a &\textrm{otherwise.}
\end{array} \right.
\end{displaymath}
Here $a$ is the normalizing constant. It is easy to verify that $g(v,k)$ satisfies {\sc condition a}.  We now find the optimal deterministic mechanism in the class of non-wasteful mechanisms. Let $(f,p)$ be a deterministic mechanism\footnote{That is, $f_2(\cdot) \in \{0,1\}$.}.  It is defined by two prices $\rho_1,\rho_2$ for types with $k =0.75,1$, respectively.  Hence,
\begin{displaymath}
(f(v,k),p(v,k)) = \left\{ \begin{array}{ll}
(0,0,0) & \textrm{if $(v < \rho_1,~k = 0.75)$ or $(v < \rho_2,~k=1)$}\\
(k,1,\rho_1) & \textrm{if $v \geq \rho_1,~k=0.75$}\\
(k,1,\rho_2) & \textrm{if $v \geq \rho_2,~k=1$.}
\end{array} \right.
\end{displaymath}
`Horizontal' incentive constraints are trivially satisfied.  Incentive constraints of the type $(v,0.75) \rightarrow (v,1)$ imply $v - \rho_1 \geq v - \rho_2$ for any $v \geq \rho_1$. Incentive constraints of the type $(v,1) \rightarrow (v,0.75)$ imply $v - \rho_2 \geq 0.75v - \rho_1$ for any $v \geq \rho_2$. Together the constraints imply $0.75 \rho_2 \leq \rho_1 \leq \rho_2$. It is easy to see that the constraint $0.75 \rho_2 \leq \rho_1$ binds. The revenue expression (since $\rho_1 \geq 0.5$ and $\rho_2 \geq \frac{2}{3}$) is given by
\begin{equation*}
-(100+\frac{16}{9})a\rho^2_1 + (100+\frac{10}{3}+1)a\rho_1
\end{equation*}
This is maximized at $\rho_1 = 0.512$ which gives a revenue of $(26.74)a$. 

Now, consider a non-deterministic mechanism $(f',p')$ given by,
\begin{displaymath}
(f'(v,k),p'(v,k)) = \left\{ \begin{array}{ll}
(0,0,0) & \textrm{if $(v < \frac{1}{2},~k = 0.75)$ or $(v < \frac{2}{3},~k=1)$}\\
(k,1,\frac{1}{2}) & \textrm{if $v \geq \frac{1}{2},~k=0.75$}\\
(k,1,\frac{11}{16}) & \textrm{if $v \geq \frac{3}{4},~k=1$}\\
(\frac{3k}{4},\frac{3}{4},\frac{1}{2}) & \textrm{if $v \in [\frac{1}{2},\frac{3}{4}),~k=1$.}
\end{array} \right.
\end{displaymath}
is incentive compatible and generates a revenue of $(26.76)a$ which is strictly higher than the optimal deterministic mechanism.

\end{example}

In proving Theorem \ref{theo:last}, we use concavity property of {\sc Condition A} (in Lemma \ref{lem:impmnt}).  This condition cannot be weakened in this proof approach. That is not to say {\sc Condition A} is necessary for the result, the next section illustrates this fact for independent types.

For the result in Theorem \ref{theo:one},  however, we can replace {\sc Condition A} with a more standard regularity condition, that $\phi(v,k)$ is strictly increasing in $v$. This is true since we only need the optimal price conditional on $k$ to be increasing (and bounded).  For a detailed comparison of the regularity condition with {\sc Condition A} see \citet{devanur_haghpanah_psomas_2020}, Section 6.1.

\subsubsection{Independent types}

The following proposition describes the optimal mechanism when the value and ratio random variables are independent. This result does not require any other condition on the type distribution. In particular, independence allows for {\sc Condition A} to be violated.

\begin{prop}\label{prop:ind}
If $g(v,k) = g_{v}(v)g_{k}(k)$, then the following {\sc posted price} mechanism is optimal,
\begin{align*}
\big(f(v,k),p(v,k)\big) &= 
   \begin{cases}
   (k,1,p^*)  \qquad v \geq p^*\\
   (0,0,0)  \qquad \text{otherwise}
   \end{cases}\\
\text{where }p^* \text{ is any } p \text{ that maximizes } p \big(1-G_{v}(p)\big)
\end{align*}
\end{prop}
\begin{proof}
We solve the reduced problem by ignoring constraints (C2) and (C3); this can be written as:

Using $g(v|k) = g_v(v)$ we rewrite (O) as,
\begin{align*}
\max_{f_2:V\times K \rightarrow [0,1]} \int_{0}^{1}\Bigg[\int_{0}^{1} \big[v - \frac{1-G_{v}(v)}{g_{v}(v)}\big] g_{v}(v) f_2(v,k) dv \Bigg] g_{k}(k) dk. \tag{O}\\
f_2(v,k) \leq f_2(v',k)~\text{for all}~v < v',k\tag{C1}\label{eq:mon}.
\end{align*}

We first maximize the objective function point-wise for each $k$, while satisfying the constraint for that $k$. To that end, fix some $k$, and observe that maximizing the term inside large bracket along with the monotonocity constraint is the same as in the standard Myerson's problem for a general distribution. Therefore, the solution of $f_2$, as described in \citet{Myerson81}, is a step function as follows,
\begin{align*}
f_2(v,k) &= 
   \begin{cases}
   1  \qquad v \geq \rho^*\\
   0  \qquad \text{otherwise}
   \end{cases}\\
\rho^* &\text{ is any }\rho \text{ that maximizes } \rho\big(1-G_{v}(\rho)\big)
\end{align*}
Since we have picked an arbitrary $k$, and this allocation function is independent of $k$, the point-wise maximization must yield a {\sc posted price} mechanism. We need to verify that the constraints (C2) and (C3) are also satisfied. But since we have shown in Proposition \ref{prop:rdic} that a {\sc posted price} mechanism is IC mechanism; this fact together with Proposition \ref{prop:ICC} implies constraints (C2) and (C3) are satisfied.
\end{proof}

\section{Concluding Remarks}
Often, models in multidimensional are intractable, even in the two-dimensional case. Even if some models are tractable, it is hard to derive a reduced-form solution for the optimal mechanism. This paper considers a natural two-dimensional private information model with a ‘separation’ in the dimensions. While one dimension captures the value of the bundle, the other represents the ratio of consumption. This feature helps us solve the problem and provide a reduced-form solution that is simple and intuitive. The {\sc posted price} mechanism can be described by one parameter and involves a finite menu of outcomes. While the {\sc ratio-dependent posted price} mechanism involves a potentially infinite size of the menu, it has a simple feature of allocating the primary good fully and the secondary good in the desired ratio.

There are three main directions one could extend this work. First, to explore results in a broader class of distributions, and identifying non-wasteful mechanisms beyond {\sc ratio-dependent posted price} mechanism. Second, to analyze a multi-good perfect complements model. Third, to consider a scenario in which multiple agents compete for the same pair of complementary goods.

\bibliographystyle{ecta}
\bibliography{compl}

\newpage

\appendix

\section{Appendix: Omitted Proofs}
\label{sec:appa}

\subsection{Proof of the Proposition \ref{prop:rdic}.}\label{subsec:prop1}
\begin{proof}
Consider a {\sc ratio-dependent posted price} mechanism $(f,p)$ defined by a function $\psi$. We first show that $(f,p)$ is IR. For any type $(v,k)$,
\begin{displaymath}
U_{(v.k)}(f(v,k),p(v,k)) = \left\{ \begin{array}{ll}
0 & \textrm{if $v \le \psi(k)$}\\
v - \psi(k) & \textrm{otherwise.}
\end{array} \right.
\end{displaymath}

Clearly, $U_{(v.k)}(f(v,k),p(v,k)) \geq 0$ and hence $(f,p)$ is IR. We now show that $(f,p)$ is IC. Without loss of generality, consider any two representative types $(v,k), (v',k')$ such that $k' \geq k$. Note that $\frac{k}{k'}\psi(k') \leq \psi(k) \leq \psi(k')$.\\

\noindent \underline{$(v,k) \rightarrow (v',k').$} Note that $\big(f(v',k'),p(v',k')\big)$ is either $(0,0,0)$ or $(k',1,\psi(k'))$, we only need to check deviation to the latter outcome because IR implies $(v,k)$ does not deviate to a type with outcome $(0,0,0)$. We check deviation to the outcome $(k',1,\psi(k'))$ in two cases.\\
{\it Case 1: $v \leq \psi(k)$.} $(v,k)$ has no incentive to deviate to $(v',k')$ because
$$U_{(v,k)}(0,0,0) = 0 \geq v - \psi(k) \geq v - \psi(k') = v \min\{\frac{k'}{k},1\} - \psi(k') = U_{(v,k)}(k',1,\psi(k')).$$
{\it Case 2: $v > \psi(k)$.}
$$U_{(v,k)}(k,1,\psi(k)) = v - \psi(k) \geq v - \psi(k') =  v \min\{\frac{k'}{k},1\} - \psi(k') = U_{(v,k)}(k',1,\psi(k')),$$
The second inequality in the first case and the inequality in the second case come from the fact that $\psi(k) \leq \psi(k')$. This means $(v,k)$ has no incentive to deviate to $(v',k').$\\

\noindent \underline{$(v',k') \rightarrow (v,k).$} Again we only need to check $(v',k')$ deviating to the outcome $(k,1,\psi(k))$.\\
{\it Case 1: $v' \leq \psi(k')$.}
$$U_{(v',k')}(0,0,0) = 0 \geq \psi(k') \frac{k}{k'} - \psi(k) \geq v' \frac{k}{k'} - \psi(k) = v' \min\{\frac{k}{k'},1\} - \psi(k) = U_{(v',k')}(k,1,\psi(k)).$$
The first inequality comes from the condition that $\psi(k) \geq  \frac{k}{k'}\psi(k')$. \\
{\it Case 2: $v' > \psi(k')$.}
$$U_{(v',k')}(k',1,\psi(k')) = v' - \psi(k') \geq v'\frac{k}{k'} - \psi(k) =  v' \min\{\frac{k}{k'},1\} - \psi(k) = U_{(v',k')}(k,1,\psi(k)).$$
The inequality comes from the following argument. $\psi(k) \geq  \frac{k}{k'}\psi(k')$ implies $\psi(k') (1 - \frac{k}{k'}) \geq \psi(k') - \psi(k)$. Since $v' > \psi(k')$, we have $v' (1 - \frac{k}{k'}) \geq \psi(k') - \psi(k)$. Rearranging the terms we get the inequality.
\end{proof}

\subsection{Proof of Example \ref{ex:1}}\label{subsec:ex1}
We solve the following reduced optimal program which uses only a subset of the constraints.
$$\max_{\{(a_i,b_i,t_i)\}_{i \in \{1,2,3\}}} (10t_1+ t_2+t_3)a \textrm{ subject to}$$
\begin{align}  
    \label{ir}&IR~(v_3,k_3):& t_3                                         &\le \min\{ a_3, b_3\}\\
        \label{ic1to3}&(v_1,k_1) \rightarrow (v_3,k_3):&    t_1&\le t_3 + \min\{4a_1,b_1 \}-\min\{4 a_3, b_3\}   \\
         \nonumber&(v_2,k_2) \rightarrow (v_1,k_1),(v_2,k_2) \rightarrow (v_3,k_3):& t_2&\le 6\min \{ 2a_2, b_2\} \\
         \label{ic2to1}&&&- \max \big\{  6\min \{2 a_1, b_1\}-t_1, 6\min \{2 a_3, b_3\}-t_3 \big\}\\
    \label{all}&Feasibility:&~& a_i \in [0,1],~ b_i \in \{0,1\}
\end{align}
Constraint \ref{ic2to1} is obtained by combining $(v_2,k_2) \rightarrow (v_1,k_1)$ and $(v_2,k_2) \rightarrow (v_3,k_3)$. Observe that we can increase $t_3, t_1, t_2$ in constraints \ref{ir},\ref{ic1to3},\ref{ic2to1} (sequentially in that order) making them binding. This process does not violate the other constraints. Hence, we re-write (ignoring the normalizing constant $a$) the objective function to get,
\begin{align*} 
&6 \min \{ 2a_2, b_2\}+ 10 \Big(\min \{ a_3, b_3\}  + \min \{ 4a_1, b_1\}- \min \{ 4a_3,b_3\}\Big)+  \min\{ a_3, b_3\} \\
&- \max \Big\{6 \min \{ 2a_1, b_1\} - \min\{ a_3, b_3\} - \min \{ 4a_1, b_1\} +\min \{ 4a_3, b_3\}, 6 \min \{ 2a_3, b_3\}- \min \{ a_3, b_3\}\Big\}
\end{align*}
The first term is independent of the others, hence we set $(a_2,b_2) = (0.5,1)$ maximizing the corresponding revenue. If $b_1 = 0$, the contribution to the revenue from the corresponding terms is $0$. If $b_1 = 1$ and $a_1 < 0.25$ increasing $a_1$ by $\epsilon$ increases the revenue by at least $40\epsilon - 12\epsilon + 4\epsilon$. If $b_1 =1$ and $a_1 > 0.25$, decreasing $a_1$ weakly increases the revenue. Hence, we get $a_1 = 0.25$ when $b_1 = 1$. Moreover, this gives strictly positive revenue from the corresponding terms. Therefore, we have $(a_1,b_1) = (0.25,1)$ in the optimal solution to the reduced problem.

Now, we substitute the values of $(a_1,b_1),(a_2,b_2)$ and consider two cases depending on which of the two terms inside the $\max$ expression is greater,

\noindent{\bf Case 1:} $2 - \min\{ a_3, b_3\} +\min \{ 4a_3, b_3\} \geq 6 \min \{ 2a_3, b_3\}- \min \{ a_3, b_3\}$. Collecting the $(a_3,b_3)$ terms we optimize $12 \min \{a_3,b_3\} - 11 \min \{4a_3,b_3\}$ subject to $6\min \{2a_3,b_3\} - \min\{4a_3,b_3\} \leq 2$ (condition of the case). Revenue from these terms is $0$ when $b_3 = 0$, the constraint in this case is also satisfied. If $b_3 = 1$ and $a_3 > 0.25$, the constraint is violated. If $b_3 = 1$ and $a_3 \leq 0.25$ the objective is $-32 a_3$ which is maximized at $a_3 = 0$. Hence, we conclude that this case has optimal at $(a_3,b_3) = (0,0)$. Payments can be calculated from the binding constraints as $t_1=1,t_2=4,t_3=0$ leading to a total revenue of $14a$.

\noindent{\bf Case 2:} $2 - \min\{ a_3, b_3\} +\min \{ 4a_3, b_3\} < 6 \min \{ 2a_3, b_3\}- \min \{ a_3, b_3\}$. Collecting the $(a_3,b_3)$ terms we optimize $12 \min\{a_3,b_3\} - 10 \min\{4a_3,b_3\} - 6\min\{2a_3,b_3\}$ subject to $6\min \{2a_3,b_3\} - \min\{4a_3,b_3\} > 2$ (condition of the case). Using argument as above we conclude that $(a_3,b_3) = (1,1)$ in this case. This gives a total revenue of $12a$.

It is easy to verify that the menu $\{(0.25,1,1),(0.5,1,4),(0,0,0)\}$ satisfies the ignored IC and IR constraints. Therefore, this is the optimal mechanism with constraint $b_i \in \{0,1\}$.

\subsection{Proof of the Proposition \ref{prop:red}.}\label{subsec:prop2}
\begin{proof}
Fix an IC and IR mechanism $(f,p)$ and define $(f',p')$ as follows,
\begin{displaymath}
(f'(v,k),p'(v,k)) := \left\{ \begin{array}{ll}
 \big(f_1(v,k),\frac{f_1(v,k)}{k},p(v,k)\big) & \textrm{if $\frac{f_1(v,k)}{k} \leq f_2(v,k)$}\\
 \big(kf_2(v,k),f_2(v,k),p(v,k)\big) & \textrm{if $\frac{f_1(v,k)}{k} > f_2(v,k)$}.
\end{array} \right.
\end{displaymath}

The {\it non-wasteful} mechanism $(f',p')$ generates as much revenue as $(f,p)$. Showing that it satisfies IC and IR conditions will prove the proposition. Fix any type $(v,k)$ and to show that this type does not deviate to some other type $(u,j)$, we do this in two cases.\\

\noindent{\sc Case 1 - $\frac{f_1(u,j)}{j} \leq f_2(u,j)$ }. 
\begin{align}
U_{(v,k)}(f'(v,k),p'(v,k)) &= U_{(v,k)}(f(v,k),p(v,k))\\
						   &\geq U_{(v,k)}(f(u,j),p(u,j)) \\
						   &= v \min\{\frac{f_1(u,j)}{k},f_2(u,j)\} - p(u,j)\\
						   &\geq v \min\{\frac{f_1(u,j)}{k},\frac{f_1(u,j)}{j}\} - p(u,j)\\
						   &= U_{(v,k)}(f_1(u,j),\frac{f_1(u,j)}{j},p(u,j))\\
						   &= U_{(v,k)}(f'(u,j),p'(u,j)).
\end{align}

\noindent{\sc Case 2 - $\frac{f_1(u,j)}{j} > f_2(u,j)$}. 
\begin{align*}
U_{(v,k)}(f'(v,k),p'(v,k)) &= U_{(v,k)}(f(v,k),p(v,k))\\
						   &\geq U_{(v,k)}(f(u,j),p(u,j)) \\
						   &= v \min\{\frac{f_1(u,j)}{k},f_2(u,j)\} - p(u,j)\\
						   &\geq v \min\{\frac{jf_2(u,j)}{k},f_2(u,j)\} - p(u,j)\\
						   &= U_{(v,k)}(jf_2(u,j),f_2(u,j),p(u,j))\\
						   &= U_{(v,k)}(f'(u,j),p'(u,j)).
\end{align*}

In both the cases, first inequality is by incentive compatibility of $(f,p)$, second inequality by the condition that defines the particular case, the first and last equations by construction of $(f',p')$, and the rest by definitions. Using first equations and the fact that $(f,p)$ is IR implies that $(f',p')$ is IR.
\end{proof}

\subsection{Proof of the Proposition \ref{prop:ICC}.}\label{subsec:propicc}

\begin{proof}
Let a mechanism $(f,p) \in \mathcal{M}$ be IC, then to show (1) and (2) fix some $k$. For any $v' > v$, consider the following IC constraints,
\begin{align*}
(v,k) \rightarrow (v',k) &\equiv vf_2(v,k) - p(v,k) \geq vf_2(v',k) - p(v',k)\\
(v',k) \rightarrow (v,k) &\equiv v'f_2(v',k) - p(v',k) \geq v'f_2(v,k) - p(v,k).
\end{align*}
After suppressing $k$ in the above inequalities notice that these are the standard one-dimensional IC constraints between two types $v,v'$. Therefore, in similar fashion to the one-dimensional problem we get (1) by adding the inequalities. For any $k$ applying \citet{Myerson81}'s revenue equivalence formula we get 
$$p(v,k) = p(0,k) + vf_2(v,k) - \int_{0}^{v}f_2(t,k)dt~\text{for all}~v$$
Applying Lemma \ref{lem:zero} to this expression we get (2).\\

To show (3) and (4) consider any $v,k' > k$. IC constraint $(v,k) \rightarrow (v,k')$ implies that,
\begin{align}
U_{(v,k)}(f(v,k),p(v,k)) &\geq U_{(v,k)}(f(v,k'),p(v,k')) \nonumber\\
\implies vf_2(v,k) - p(v,k) &\geq v\min\{\frac{f_1(v,k')}{k},f_2(v,k')\} - p(v,k') \nonumber\\
                                            & = v\min\{\frac{k'f_2(v,k')}{k},f_2(v,k')\} - p(v,k') \nonumber\\
                                            & = vf_2(v,k') - p(v,k') \nonumber\\
\implies \int_{0}^{v}f_2(t,k)dt &\geq \int_{0}^{v}f_2(t,k')dt \nonumber
\end{align}
IC constraint $(v,k') \rightarrow (v\frac{k}{k'},k)$ implies that,
\begin{align}
U_{(v,k')}(f(v,k'),p(v,k')) &\geq U_{(v,k')}(f(v\frac{k}{k'},k),p(v\frac{k}{k'},k)) \nonumber\\
\implies vf_2(v,k') - p(v,k') &\geq v\min\{\frac{f_1(v\frac{k}{k'},k)}{k'},f_2(v\frac{k}{k'},k)\} - p(v\frac{k}{k'},k) \nonumber\\
                                            & = v\min\{\frac{k}{k'}f_2(v\frac{k}{k'},k),f_2(v\frac{k}{k'},k)\} - p(v\frac{k}{k'},k) \nonumber\\
                                            & = v\frac{k}{k'}f_2(v\frac{k}{k'},k) - p(v\frac{k}{k'},k) \nonumber\\
                                            & = U_{(v\frac{k}{k'},k)}(f(v\frac{k}{k'},k),p(v\frac{k}{k'},k)) \nonumber\\
\implies \int_{0}^{v}f_2(t,k')dt &\geq \int_{0}^{v\frac{k}{k'}}f_2(t,k)dt \nonumber
\end{align}
The first equality in both the constraints uses the fact that $(f,p) \in \mathcal{M}$. The second implication uses the necessary condition (2) of this Proposition.\\

\begin{figure}[!hbt]
\centering
\includegraphics[width=3in]{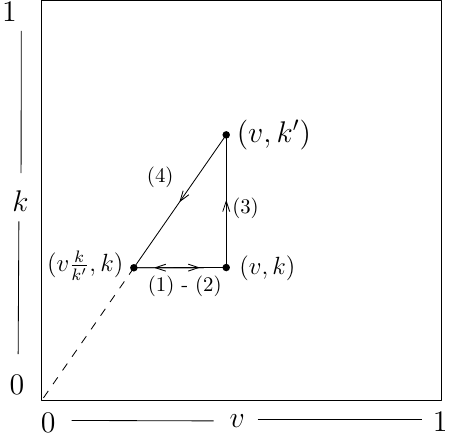}
\caption{IC Constraints}
\end{figure}

For the only if part, fix any $k$ and notice that IC constraints of the type $(v,k) \rightarrow (v',k')$ when $k = k'$ are satisfied by conditions (1) and (2) as this is equivalent to standard one-dimensional one agent model. Therefore, it is enough to show that any type $(v,k)$ does not deviate to a $(v',k')$ in the following two cases.\\
\noindent {\bf Case 1: $k' > k$.}
\begin{align*}
U_{(v,k)}(f(v,k),p(v,k)) &= vf_2(v,k)-p(v,k)\\
                                        &= \int_{0}^{v}f_2(t,k)dt - p(0,1)\\
                                        &\geq \int_{0}^{v}f_2(t,k')dt - p(0,1)\\
                                        &= vf_2(v,k')-p(v,k')\\
                                        & \geq vf_2(v',k')-p(v',k')\\
                                        & = v\min\{\frac{k'}{k}f_2(v',k'),f_2(v',k')\}-p(v',k')\\
                                        &=U_{(v,k)}(f(v',k'),p(v',k'))
\end{align*}

The second and third equation uses condition (2).  The second inequality is from IC constraint $(v,k') \rightarrow (v',k')$ which in turn come from conditions (1) and (2) as argued already. The first inequality is from condition (3).

\noindent {\bf Case 2: $k' < k$.}
\begin{align*}
U_{(v,k)}(f(v,k),p(v,k)) &= vf_2(v,k)-p(v,k)\\
                                        &= \int_{0}^{v}f_2(t,k)dt - p(0,1)\\
                                        &\geq \int_{0}^{v\frac{k'}{k}}f_2(t,k')dt - p(0,1)\\
                                        &= v\frac{k'}{k}f_2(v\frac{k'}{k},k')-p(v\frac{k'}{k},k')\\
                                        & \geq v\frac{k'}{k}f_2(v',k')-p(v',k')\\
                                        & = v\min\{\frac{k'}{k}f_2(v',k'),f_2(v',k')\}-p(v',k')\\
                                        &=U_{(v,k)}(f(v',k'),p(v',k'))
\end{align*}

The second and third equation uses condition (2), the second inequality is from IC constraint $(v\frac{k'}{k},k') \rightarrow (v',k')$ which in turn come from conditions (1) and (2) as argued already. The first inequality is from condition (4).
\end{proof}

\subsection{Proof of Lemma \ref{lem:impmnt}}\label{proof:lem3}
It is straightforward to see that constraint (C1) is satisfied. For (C2), observe that, for any $(v,k)$,
\begin{equation}\label{eq:new}
\int_{0}^{v}f_2'(t,k)dt = \left\{ \begin{array}{ll}
0 & \textrm{if $v \le 1-\int_{0}^{1}f_2(t,k)dt$}\\
v- 1+\int_{0}^{1}f(t,k)dt& \textrm{otherwise.}
\end{array} \right.
\end{equation} 
Fix any $k' > k$, and since $(f,p)$ satisfies constraint (C2) we have, 
\begin{equation}\label{eq:sc2}
\int_{0}^{1}f_2(t,k')dt \leq \int_{0}^{1}f_2(t,k)dt.
\end{equation}

If $v \leq 1-\int_{0}^{1}f_2(t,k')dt$, then $\int_{0}^{v}f_2'(t,k')dt = 0 \leq \int_{0}^{v}f_2'(t,k)dt$, as $f_2'(v,k) \geq 0~\forall (v,k)$. \\

Else if $v > 1-\int_{0}^{1}f_2(t,k')dt$, then $v > 1-\int_{0}^{1}f_2(t,k)dt$ by equation \ref{eq:sc2}. Therefore, $\int_{0}^{v}f_2'(t,k')dt = v -1 + \int_{0}^{1}f_2(t,k') \leq v -1 + \int_{0}^{1}f_2(t,k) = \int_{0}^{v}f_2'(t,k)dt$. The inequality is by equation \ref{eq:sc2}. The equations are by expression \ref{eq:new}.

Now we show that $(f',p')$ generates weakly more expected revenue than $(f,p)$. Fix any $k$. Denote $\beta_{(f,p,k)} := 1-\int_{0}^{1}f_2(t,k)dt$ and consider the difference in expected revenue of the two mechanisms,
\begin{align*}
\int_{0}^{1}\phi(v,k)g(v|k)\big(f_2'(v,k) - f_2(v,k)\big)dv &= \int_{\beta_{(f,p,k)}}^{1}\phi(v,k)g(v|k)\big(f_2'(v,k)-f_2(v,k)\big)dv\\
                                                                                                   &\qquad - \int_{0}^{\beta_{(f,p,k)}}\phi(v,k)g(v|k)f_2(v,k)dv\\
                                                                                                   &\geq \phi(\beta_{(f,p,k)},k)g(\beta_{(f,p,k)}|k)\int_{\beta_{(f,p,k)}}^{1}\big(f_2'(v,k)-f_2(v)\big)dv\\
                                                                                                   &\qquad - \phi(\beta_{(f,p,k)},k)g(\beta_{(f,p,k)}|k)\int_{0}^{\beta_{(f,p,k)}}f_2(v,k)dv\\
                                                                                                   &= \phi(\beta_{(f,p,k)},k)g(\beta_{(f,p,k)}|k) \big(\int_{0}^{1}(f_2'(v,k)-f_2(v,k))dv\big)\\
                                                                                                   &=0
\end{align*}
The equations use the definition of $(f',p')$ and rearranging of terms, the inequality is from the fact that $\phi(v,k)g(v|k)$ is increasing. Since we have shown this for an arbitrary $k$ therefore expected revenue from $(f',p')$ is greater(weakly) than $(f,p)$.

\subsection{Ironing for Theorem \ref{theo:last}.}\label{subsec:ironing}

Fix any mechanism $(f,p)$ described in \ref{eq:theo} and note that {\sc Condition B} implies $\phi^{-1}_{k}(0) > \phi_{k'}^{-1}(0)~\forall k' > k$. Consider the following three mutually exclusive and exhaustive cases:
\begin{enumerate}
\item $\rho(1) \leq \phi_{1}^{-1}(0)$. Consider the following mechanism $(f',p')$ defined by,
\begin{displaymath}
f_2'(v,k) = \left\{ \begin{array}{ll}
0 & \textrm{if $v \le \rho(1)$}\\
1 & \textrm{otherwise.}
\end{array} \right.
\end{displaymath}
Fix any $k$, note that $\rho(k) \leq \rho(1)$. If $v \leq \rho(k)$ or $v > \rho(1)$ then $f_2'(v,k) = f_2(v,k)$. $\rho(1) \leq \phi_{1}^{-1}(0) \leq \phi_{k}^{-1}(0)$ implies $\phi(\rho(1),k)g(\rho(1)|k) \leq \phi(\phi_{k}^{-1}(0),k)g(\phi_{k}^{-1}(0)|k) = 0$ since $\phi(v,k)g(v|k)$ is increasing in $v$.  This implies $\phi(v,k)g(v|k) \leq 0$ for all $v \leq \rho(1)$. Hence, $\int_{\rho(k)}^{\rho(1)}f_2(v,k)\phi(v,k)g(v|k)dv \leq 0$. Noticing $\int_{\rho(k)}^{\rho(1)}f_2'(v,k)\phi(v,k)g(v|k)dv = 0$ by construction we conclude revenue in $(f',p')$ is more than that of $(f,p)$.\\

With some abuse of notation , we use $\rho(0^{+})$ to denote $\lim_{k \to 0^{+}}\rho(k)$, and $\phi_{0^{+}}^{-1}(0)$ to denote $\lim_{k \to 0^{+}}\phi_{k}^{-1}(0)$.

\item $\rho(1) > \phi_{1}^{-1}(0)$ and $\rho(0^{+}) < \phi_{0^{+}}^{-1}(0)$. $\rho$ is increasing in $k$, and $\phi^{-1}_{k}(0)$ is strictly decreasing in $k$ and continuous, hence the function $\rho(k) - \phi^{-1}_{k}(0)$ is strictly increasing (and continuous a.e.). Therefore, there exists a unique $k^*$ such that $\rho(k) > \phi_{k}^{-1}(0)~\forall k > k^*$ and $\rho(k) < \phi_{k}^{-1}(0)~\forall k < k^*$. Let $v^* := \rho(k^*)$. Define a {\sc posted price} mechanism $(f',p')$ as follows,
\begin{displaymath}
f_2'(v,k) = \left\{ \begin{array}{ll}
0 & \textrm{if $v \le v^*$}\\
1 & \textrm{otherwise.}
\end{array} \right.
\end{displaymath}
We show that $(f',p')$ generates more expected revenue than $(f,p)$ for every $k$ in two following cases.
\begin{enumerate}
\item Fix any $k > k^*$. Note that $v^* \leq \rho(k)$. If $v \leq v^*$ or $v > \rho(k)$ then $f_2'(v,k) = f_2(v,k)$. Since $\phi_k^{-1}(0) \leq \phi_{k^*}^{-1}(0) = v^*$ and $\phi(v,k)g(v|k)$ increasing in $v$ we have $\phi(v,k)g(v|k) > 0$ for all $v > v^*$. Therefore, $\int_{v^*}^{\rho(k)}\big(f_2'(v,k)-f_2(v,k)\big)\phi(v,k)g(v|k)dv \geq 0$ since $f_2'(v,k) = 1$ in this range.
\item Fix any $k < k^*$. Note that $v^* \geq \rho(k)$. If $v > v^*$ or $v \leq \rho(k)$ then $f_2'(v,k) = f_2(v,k)$. Since $\phi_k^{-1}(0) \geq \phi_{k^*}^{-1}(0) = v^*$ and $\phi(v,k)g(v|k)$ increasing in $v$ we have $\phi(v,k)g(v|k) < 0$ for all $v < v^*$. Therefore, $\int_{\rho(k)}^{v^*}\big(f_2'(v,k)-f_2(v,k)\big)\phi(v,k)g(v|k)dv \geq 0$ since $f_2'(v,k) = 0$ in this range.
\end{enumerate}
\item $\rho(0^{+}) \geq \phi_{0^{+}}^{-1}(0)$. Consider the following mechanism $(f',p')$ defined by,
\begin{displaymath}
f_2'(v,k) = \left\{ \begin{array}{ll}
0 & \textrm{if $v \le \rho(0^{+})$}\\
1 & \textrm{otherwise.}
\end{array} \right.
\end{displaymath}
Fix any $k$ and note that $\rho(k) \geq \rho(0^{+})$. If $v \leq \rho(0^{+})$ then $f_2'(v,k) = f_2(v,k)$. Since $\phi_k^{-1}(0) \leq \phi_{0^{+}}^{-1}(0) \leq \rho(0^{+})$ and $\phi(v,k)g(v|k)$ increasing in $v$ we have $\phi(v,k)g(v|k) > 0$ for all $v > \rho(0^{+})$. Therefore, $\int_{\rho(0^{+})}^{1}\big(f_2'(v,k)-f_2(v,k)\big)\phi(v,k)g(v|k)dv \geq 0$ since $f_2'(v,k) = 1$ in this range.
\end{enumerate}

In each of the above three cases, we have shown that the revenue is higher in a {\sc posted price} mechanism for an arbitrary $k$. Therefore, a {\sc posted price} mechanism is optimal in the reduced problem we considered. This also implies it is the optimal mechanism since we have shown that a {\sc posted price} mechanism satisfies all the constraints, including the ignored constraint (C3).

\end{document}